\documentclass[12pt]{article}

\newtheorem{thm}{Theorem}

\newtheorem{lemma}{Lemma}

\newenvironment{proof}[1][Proof]{\textbf{#1.} }{\ \rule{0.5em}{0.5em}}

\usepackage{amsmath,amsfonts,amssymb}
\begin{document}
\title{A complexity dichotomy for the dominating set problem}
\author{D.S. Malyshev\footnote{National Research University Higher School of Economics, 25/12 Bolshaya Pecherskaya
Ulitsa, 603155, Nizhny Novgorod, Russia}}
\date{dsmalyshev@rambler.ru}

\maketitle

\begin{abstract}

We completely determine the complexity status of the dominating set
problem for hereditary graph classes defined by forbidden
induced subgraphs with at most five vertices.

\end{abstract}

\section{Introduction}

A \emph{coloring} is an arbitrary mapping
of colors to vertices of some graph. A graph coloring is said to be
\emph{proper} if no pair of adjacent vertices have the same color. The
\emph{chromatic number $\chi(G)$} of a graph $G$ is the
minimal number of colors in proper colorings of $G$. The \emph
{coloring problem}, for a given graph and a number $k$, is to
determine whether its chromatic number is at most $k$ or not. The
\emph{vertex $k$-colorability problem} is to verify whether vertices of a
given graph can be properly colored with at most $k$ colors. The edge
$k$-colorability problem is defined by analogy.

An \emph{independent set} and a \emph{clique} of a graph are sets
of pairwise non-adjacent and adjacent vertices, respectively.
The \emph{independent set problem} is to determine whether a given graph
contains an independent set with a given number of elements. The clique
problem is defined by analogy.

For a graph $G$, a subset $V'\subseteq V(G)$ \emph{dominates} $V''\subseteq V(G)$ if
each vertex of $V''\setminus V'$ has a neighbor in $V'$. A \emph{dominating set} of a graph $G$
is a subset dominating all its vertices. The size of a minimum dominating set of $G$
is said to be the \emph{domination number} of $G$ denoted by $\gamma(G)$. For a graph $G$
and a number $k$, the \emph{dominating set problem} is to decide whether
$\gamma(G)\leq k$ or not.

A \emph{class} is a set
of simple unlabeled graphs. A class of graphs is \emph{hereditary} if it is closed under deletion of vertices. It is
well-known that any hereditary (and only hereditary) graph class
${\mathcal X}$ can be defined by a set of its forbidden induced
subgraphs ${\mathcal Y}$. We write ${\mathcal X}=Free({\mathcal Y})$
in this case, and the graphs in
${\mathcal X}$ are said to be \emph{${\mathcal Y}$-free}. If ${\mathcal Y}=\{G\}$, then we will
write ``$G$-free'' instead of ``$\{G\}$-free''. If a hereditary class can be defined by a finite set
of forbidden induced subgraphs, then it is said to be \emph{finitely
defined}.

The coloring problem for $G$-free graphs is polynomial-time solvable
if $G$ is an induced subgraph of $P_4$ or $P_3+K_1$, and it is
NP-complete in all other cases \cite{KKTW01}. A similar result is
known for the dominating set problem. Namely, the problem is
polynomial-time solvable for $Free(\{G\})$ if $G=P_i+O_k$, where
$i\leq 4$ and $k$ is arbitrary, and it is NP-complete for all other choices of $G$ \cite{K90}.
The situation for the vertex $k$-colorability problem is not clear, even when only one induced subgraph
is forbidden. The complexity of the vertex 3-colorability problem is known
for all classes of the form $Free(\{G\})$ with $|V(G)|\leq 6$
\cite{BGPS12}. A similar result for $G$-free graphs with $|V(G)|\leq
5$ was recently obtained for the vertex 4-colorability problem
\cite{GPS13}. On the other hand, for fixed $k$, the complexity status of the
vertex $k$-colorability problem is open for $P_7$-free graphs $(k=3)$, for
$P_6$-free graphs $(k=4)$, and for $P_3+P_2$-free graphs
$(k=5)$.

The independent set problem is polynomial-time solvable for a hereditary class
defined by forbidden induced subgraphs with at most five vertices if and only if
a forest is one of the subgraphs, unless $P=NP$ \cite{LVV14,LM04}. A similar complete complexity
dichotomy was obtained in \cite{M14_3} for the edge 3-colorability problem. For the coloring problem, a complete classification for
pairs is open, even if forbidden induced subgraphs have at most four vertices.
Although, the complexity is known for some such pairs  \cite{GP13,LM14,M14_2,M15_1,MO15}.

We present a complete dichotomy for the dominating set problem in the family of hereditary
classes defined by forbidding induced subgraphs with at most five vertices in the paper.

\section{Notation}

We use the standard notation $P_n,O_n,K_n$ for a
simple path, an empty graph, and a complete graph with
$n$ vertices, respectively. A graph $K_{p,q}$ is a complete
bipartite graph with $p$ vertices in the first part and $q$ in
the second. A graph $fork$ is obtained from a $K_{1,3}$ by subdividing
an arbitrary its edge. A graph $orb$ is obtained from
a $K_4$ by adding a new vertex and an edge connecting the added vertex
to one vertex of a $K_4$. Similarly, a graph $sinker$ is obtained by
adding a vertex and two edges incident to the new vertex and two vertices of a $K_4$.
A graph $bull$ is obtained from a $P_5$ by connecting the second and fourth its vertices by an edge.
A graph $cricket$ is obtained from a $K_3$ by adding two vertices and two edges incident to
the new vertices and the same vertex of $K_3$. Graphs $dart$ and $kite$ are obtained from
a $K_4$ minus an edge by adding a vertex and an edge incident to the new vertex and to a degree
three or a degree two vertex, respectively. A graph $gem$ is obtained from a $P_4$ by adding
a new vertex and four edges incident to the new vertex and all vertices of $P_4$. A graph $hammer$
is obtained from a $fork$ by adding a new edge incident to two leaves adjacent to the degree three vertex.

A formula $N(x)$ denotes the neighborhood of a vertex $x$. A \emph{sum} $G_1+G_2$ is the
disjoint union of $G_1$ and $G_2$ with non-intersected sets of vertices. A \emph{product}
$G_1\times G_2$ of graphs with non-intersected sets of vertices is a graph $(V(G_1)\cup V(G_2), E(G_1)\cup E(G_2)\cup \{(v,u)|~v\in V(G_1),u\in V(G_2)\})$.
For a graph $G$ and $V'\subseteq V(G)$, a graph $G[V']$ is the subgraph of $G$ induced by $V'$.

We refer to textbooks in graph theory for graph terminology undefined here.

\section{Boundary graph classes for the dominating set problem}

A large number of results on polynomial-time solvability and NP-complete-ness
has been accumulated for many graph problems under various restrictions for
graph classes. When considering representative families of graph classes, one could set
more general problems than the complexity analysis of some concrete graph
problem for a given class of graphs. How to classify classes in a family
with respect to the computational complexity of a considered graph problem?
When does a difficult problem became easy? Is there a boundary separating
``easy'' and ``hard'' instances? The aim of this section is to present some tool
based on the notion of a boundary graph class giving a complete complexity
dichotomy in the family of finitely defined graph classes.

To solve any of the mentioned problems in the family of hereditary classes, a natural idea coming to mind is to consider a
phase transition between easy and hard hereditary classes under some natural statements of the easiness and hardness.
We use the following formal definitions.
For a given NP-complete graph problem $\Pi$, a hereditary class is said to be
\emph{$\Pi$-easy} if $\Pi$ can be polynomially solved for its graphs. A hereditary
class is \emph{$\Pi$-hard} if $\Pi$ is NP-complete for it. Unfortunately, the phase transition
approach seems to be unsuccessful.

Maximal $\Pi$-easy and minimal $\Pi$-hard classes are natural
boundary elements in the lattice of hereditary classes. It turns out that
the boundary may be absent at all. First, there are no maximal $\Pi$-easy classes, as
any $\Pi$-easy class ${\mathcal X}$ can be extended by adding a graph $G\not \in {\mathcal X}$ and
all proper induced subgraphs of $G$. Clearly, the resultant class is also $\Pi$-easy. Second,
minimal hard classes exist for some problems and do not exist for some others.
For a given graph and a positive length function on its edges, the \emph{travelling salesman problem} is to check
whether the minimum length of its Hamiltonian cycles is at most a given number or not.
It is NP-complete in the class of all complete graphs. Each proper hereditary subclass of the class is finite.
Hence, the class is a minimal hard case for the problem. On the other hand, for the
vertex and edge variants of the $k$-colorability problem, any hard class contains a proper hard subclass.
Indeed, if ${\mathcal Y}$ is a minimal hard case for the problem, then it must contain a graph $H$ that cannot
be properly colored in $k$ colors. Therefore, ${\mathcal Y}\setminus Free(\{H\})$ contains only graphs that
also cannot be properly colored in $k$ colors. There is a trivial polynomial-time algorithm to test whether a given
graph in ${\mathcal Y}$ belongs to ${\mathcal Y}\cap Free(\{H\})$. Hence, ${\mathcal Y}\cap Free(\{H\})$ must be
hard for the problem, and we have a contradiction. The phenomena of the absence of the boundary
we just considered was noticed in \cite{M09_2}.

So, to classify hereditary classes, we have to take into account that
the sets of easy and hard classes can be open with respect to the inclusion relation. In other words, there
may be infinite monotonically decreasing sequences of hard classes. Intuitively, the limits of such chains
should play a special role in the analysis of the complexity. This observation leads to the notion of a
boundary graph class. A class ${\mathcal X}$ is \emph{$\Pi$-limit} if there is
an infinite sequence ${\mathcal X_1}\supseteq {\mathcal X_2}
\supseteq \ldots$ of $\Pi$-hard classes such that ${\mathcal
X}=\bigcap\limits_{k=1}^{\infty}{\mathcal X_k}$. A $\Pi$-limit class that is minimal under inclusion is said to be \emph{$\Pi$-boundary}.
The following theorem shows the significance of the boundary class notion.

\begin{thm} $\cite{A04,ABKL07}$
A finitely defined class is $\Pi$-hard if and only if it includes some
$\Pi$-boundary class.
\end{thm}

The theorem says that knowledge of all $\Pi$-boundary classes gives
a complete complexity dichotomy in the family of finitely defined classes. Moreover,
its corollary is a ``zero-one law'' claiming that there is no finitely defined classes with
an intermediate (i.e., distinct to polynomial-time solvability and NP-completeness) complexity.
One more interesting fact is that there is a boundary class for each NP-complete graph problem,
as the set of all graphs is finitely defined.

The notion was originally introduced by V.E. Alekseev for the independent set problem
\cite{A04}. It was later applied for the dominating set problem
\cite{AKL04}. Nowadays, boundary classes are known for several algorithmic
 graph problems \cite{A04,ABKL07,AKL04,KLMT11,LP13,M09_2,M13,M14_1,M15_2}.

Assuming $P\neq NP$, four concrete graph classes are known to be boundary for the dominating set problem \cite{AKL04,M15_2}.
The first of them is ${\mathcal S}$. It constitutes all forests with at
most three leaves in each connected component. The second one is
${\mathcal T}$ which is the set of line graphs of graphs in
${\mathcal S}$. To define two remaining classes, we need to define
two operators acting on graphs.

For a graph $G=(V,E)$,
a graph $Q(G)$ has vertex set $V\cup E$ and edge set $
\{(v_i,v_j)|~v_i,v_j\in V\}\cup \{(v,e)|~v\in V,e\in E,v~\text{is incident to}~e\}$.
A class ${\mathcal Q}$ is the set $\{G|~\exists H\in {\mathcal S},G=Q(H)\}$ plus
the set of all induced subgraphs of all its graphs. Let $G=(V,E)$ be a graph having degrees of vertices at most three. Let $V'$ be the set of degree three vertices of $G$ and $V''\triangleq V(G)\setminus V'$. We define a graph $Q^*(G)$ as follows. The set $V(Q^*(G))$ coincides with
$V''\cup E$. A vertex $x\in V'$ is incident to edges $e_1(x),e_2(x),e_3(x)$ in the graph $G$.
The set $E(Q^*(G))$ coincides with $\{(v_i,v_j)|~v_i,v_j\in V''\}\cup \{(v,e)|~v\in V'',e\in E,v~\text{is incident to}~e\}\cup
\bigcup\limits_{x\in V'}\{(e_1(x),e_2(x)),(e_1(x),e_3(x)),(e_2(x),e_3(x))\}$. A class ${\mathcal Q^*}$ is the set $\{G|~\exists H\in {\mathcal S},G=Q^*(H)\}$ plus
the set of all induced subgraphs of all its graphs.

Taking into account Theorem 1, a necessary condition for a finitely defined class to be an easy case for the dominating set problem
is not to include each of the classes ${\mathcal S},{\mathcal T},{\mathcal Q},{\mathcal Q^*}$, unless $P=NP$. Sometimes, they give
a criterion. The following result was obtained in \cite{M15_2}.

\begin{thm}
Let ${\mathcal Y}$ be a set of graphs with at most five vertices. If $P_5\in {\mathcal Y}$, then
the dominating set problem is polynomial-time solvable for $Free({\mathcal Y})$ if $Free({\mathcal Y})\nsupseteq {\mathcal Q}$,
otherwise it is NP-complete for $Free({\mathcal Y})$.
\end{thm}

In this paper, we extend Theorem 2 by giving a criterion for all possible subsets of forbidden
induced subgraphs with at most five vertices. Namely, such a class is hard for the problem
if it includes ${\mathcal S}$ or ${\mathcal T}$ or ${\mathcal Q}$, otherwise it is easy.

\section{The basic idea and the first steps of its implementation}

An \emph{independent dominating set} of a graph $G$ is a subset $V'\subseteq V(G)$ which
is an independent set of $G$ and a dominating set of $G$, simultaneously. The
size of a minimum independent dominating set of $G$ is said to be the \emph{independent domination number} of
$G$ denoted by $i(G)$.

Let $G$ be a connected $P_3+P_2$-free graph, $x$ and $y$ be its adjacent vertices.
Let $G_{xy}$ be the induced subgraph of $G$ obtained by deleting $x$ and $y$ simultaneously. Its
vertex set can be partitioned into two parts $A_{xy}$ and $B_{xy}$, where $A_{xy}\triangleq \{z\in V(G_{xy})|~ z\in N(x)\cup N(y)\}$.
Let $\gamma'(G_{xy})$ be the minimum cardinality of subsets of $V(G_{xy})$ that dominate $B_{xy}$. Clearly,
$\gamma(G)=\min(i(G),2+\min\limits_{xy\in E(G)}\gamma'(G_{xy}))$.

The independent domination number can be computed in polynomial time for $P_3+P_2$-free graphs \cite{LMP15}. Therefore, to show polynomial-time
solvability of the dominating set problem in a subclass ${\mathcal X}\subseteq Free(\{P_3+P_2\})$, it is sufficient to
compute $\gamma'(G_{xy})$ in time bounded by a polynomial on $|V(G)|$ for every $G\in {\mathcal X}$ and each edge $xy$ of $G$.
This reduction is our basic idea.

Clearly, $G[B_{xy}]$ is $P_3$-free, i.e. it is the disjoint union of complete graphs. If a vertex $v \in B_{xy}$ has no neighbors in $A_{xy}$,
then any dominating set of $G$ must contain an element of the clique of $G[B_{xy}]$ containing $v$. Removing this clique produces an induced
subgraph $H$ of the graph $G$ such that $\gamma'(H_{xy})=\gamma'(G_{xy})-1$. This is why we shall always assume that each element
of $B_{xy}$ has a neighbor in $A_{xy}$, since computing $\gamma'(G_{xy})$ can be polynomially reduced to this case. Let
$\gamma''(G_{xy})$ be the minimum cardinality of subsets of $A_{xy}$ that dominate $B_{xy}$, and let $k_{xy}$ be the number
of connected components of $G[B_{xy}]$.

\begin{lemma}
If $\gamma(G)\geq 4$, then $\gamma'(G_{xy})=\min(\gamma''(G_{xy}),k_{xy})$.
\end{lemma}

\begin{proof} A connected component of $G[B_{xy}]$ is \emph{non-trivial} if it has at least two vertices. Otherwise, it is
said to be \emph{trivial}. As $\gamma(G)\geq 4$, then $G[B_{xy}]$ has at least two connected components. Clearly,
$\gamma'(G_{xy})\leq \min(\gamma''(G_{xy}),k_{xy})$. Let $D_{xy}$ be a minimum set of $V(G_{xy})$ dominating $B_{xy}$. It must have
at least two elements. If this set contains no elements of $A_{xy}$, then $|D_{xy}|=k_{xy}=\gamma(G'_{xy})$.
Therefore, we may assume that $D_{xy}\cap A_{xy}\neq \emptyset$. Let us reconstruct $D_{xy}$ as follows. If it contains a vertex $b$ belonging to
a trivial component of $G[B_{xy}]$, then a new value of $D_{xy}$ becomes equal to the old one minus $\{a\}$ plus $\{b\}$, where
$a\in A_{xy}$ is an arbitrary neighbor of $b$. After this process, $D_{xy}$ is still a minimum set dominating $B_{xy}$. Let
$G_1,G_2,\ldots,G_s$ be all maximal induced subgraphs of connected components of $G[B_{xy}]$ such that every vertex of each subgraph
has no a neighbor in $D_{xy}\cap A_{xy}$. Each subgraph has exactly one vertex. Indeed, for any $i\in \overline{1,s}$, there are a vertex $z=z(i)\in D_{xy}\cap A_{xy}$
and a connected component $K=K(i)$ of $G[B_{xy}]$ such that $V(G_i)\cap V(K)=\emptyset$ and $z$ has a neighbor $z'\in V(K)$. If $|V(G_i)|\geq 2$, then
$z',z,x$ or $y$, and any two elements of $V(G_i)$ induce a $P_2+P_3$. Hence, $|V(G_i)|=|V(G_i)\cap D_{xy}|=1$ for each $i$. The element $b_i$ of
$V(G_i)\cap D_{xy}$ has a neighbor $a_i\in A_{xy}$. Therefore, $(D_{xy}\setminus \bigcup\limits_{i=1}^{s} \{b_i\})\cup \bigcup\limits_{i=1}^{s} \{a_i\}$
is a subset of $A_{xy}$ with $\gamma'(G_{xy})$ vertices dominating $B_{xy}$. Hence, $\gamma''(G_{xy})\leq \gamma'(G_{xy})$, i.e. $\gamma''(G_{xy})=\gamma'(G_{xy})$.  \end{proof}\\

Let $A'_{xy}$ be the set of those elements $z\in A_{xy}$ having a neighbor in $B_{xy}$ that $B_{xy}\setminus N(z)$ is independent,
$H^z\triangleq G\setminus (\{z\}\cup N(z)\cap B_{xy})$ for $z\in A'_{xy}$.

\begin{lemma} If $\gamma(G)\geq 4$, then $\gamma''(G_{xy})=\min\limits_{z\in A'_{xy}}\gamma''(H^{z}_{xy})+1$.
\end{lemma}

\begin{proof}
Since $\gamma(G)\geq 4$, the graph $G[B_{xy}]$ has at least two connected components.
If a minimum subset of $A_{xy}$ dominating $B_{xy}$ contains an element of $A'_{xy}$,
then $\gamma''(G_{xy})=\min\limits_{z\in A'_{xy}}\gamma''(G^{z}_{xy})+1$.
Hence, we need to assume that none element of such a subset $D_{xy}$ is an element
of $A'_{xy}$. Since $\gamma(G)\geq 4$, this subset must contain at least two elements.
To avoid an induced $P_3+P_2$, $z$ must have neighbors only in one connected component
 of $G[B_{xy}]$ for any $z\in D_{xy}$. Indeed, two adjacent elements of $B_{xy}\setminus N(z)$,
$z$, some its neighbor in $V(G[B_{xy}]), x$ or $y$ induce a $P_3+P_2$ otherwise. Let $z_1\in D_{xy}$ and
$z_2\in D_{xy}$ be vertices having neighbors $z'_1$ and $z'_2$ in distinct connected components of $G[B_{xy}]$.
If $a$ and $b$ are adjacent elements of $B_{xy}\setminus N(z_1)$, then $z'_1,a,b$ form a clique. Otherwise, $G$ is not $P_3+P_2$-free.
Clearly, $(a,z_2)\not \in E(G)$ and $(b,z_2)\not\in E(G)$. Hence, $a$ and $b$, $z'_2$, $z_2$, $x$ or
$y$ induce a $P_3+P_2$. We have a contradiction with the assumption.
\end{proof}\\

Now, let $A_{xy}$ and $B_{xy}$ mean the corresponding sets of the graph $H^{z}$ for the edge $xy$. Clearly, $B_{xy}$ is independent. We will assume that
$H^z$ is connected, $A_{xy}$ and $B_{xy}$ are non-empty, and there is no an element of $A_{xy}$ having no neighbors in $B_{xy}$. Additionally, we will also assume that $B_{xy}$ has no degree one vertices and $A_{xy}$ has no two vertices $u$ and $v$ such that $N(u)\setminus (A_{xy}\cup \{x,y\})\subseteq N(v)\setminus (A_{xy}\cup \{x,y\})$. Computing $\gamma''(H^z_{xy})$ can be easily reduced to the case in polynomial time.

\begin{lemma} For the graph $H^z$, any element of $N(x)\setminus (\{y\}\cup N(y))$ is adjacent to any element of $N(y)\setminus \{x\}$.
\end{lemma}

\begin{proof}
Assume that there are non-adjacent vertices $a\in N(x)\setminus N(y), a\neq y$ and $b \in N(y), b\neq x$. By the properties of $H^z$ above,
there are vertices $a',b'\in B_{xy}$ such that $a'\in N(a)\setminus N(b)$ and $b'\in N(b)\setminus N(a)$. Then $a',a,b',b,y$ induce a $P_3+P_2$. We have a contradiction.
\end{proof}

\section{Auxiliary results}

\subsection{Properties of irreducible graphs}

By some results of the previous section, the dominating set problem for a hereditary class ${\mathcal X}\subseteq Free(\{P_2+P_3\})$ can be polynomially
reduced to a similar-type problem for graphs in ${\mathcal X}$ whose vertex sets were partitioned into two subsets.
If $G$ is such a graph and $(A,B)$ is its partition, then we write $G\triangleq G(A,B)$. Moreover, $B$ is independent,
$B$ has no degree one vertices, $A$ has no two vertices $u$ and $v$ such that $N(u)\setminus A\subseteq N(v)\setminus A$, and
none element of $B$ is adjacent to all elements of $A$. A graph $G$
of this type is said to be \emph{irreducible}. Moreover, $A$ is split into three subsets $A_1,A_2,A_3$ such that adding vertices $x$ and
$y$ and all edges in $\{(x,x')|~x'\in A_1\cup A_3\}\cup \{(y,y')|~y'\in A_2\cup A_3\}$ to $G$ produces a graph $G'\in {\mathcal X}$.

Let $N_B(a)\triangleq\{b\in B|~(a,b)\in E(G)\}$ for a vertex $a\in A$, and let
$N_B(A')\triangleq\bigcup\limits_{a\in A'} N_B(a)$ for a subset $A'$ of $A$. Let $G^*$ be the graph obtained
from $G$ by adding the minimum possible number of edges to make $A$ to be a clique. Let $\gamma''(G)$ be the minimum cardinality of subsets
of $A$ dominating $B$. Clearly, $\gamma(G^*)=\gamma''(G)$, as there is a minimum dominating set of $G^*$ contained in $A$.

\begin{lemma}
Let $A'\triangleq\{a_1,a_2,\ldots,a_k\}$ be an independent subset of $A$ and $b_i\in N_B(a_i)\setminus \bigcup\limits_{j=1,j\neq i}^{k}N_B(a_j)$. Then
each element of $N'\triangleq N_B(A')\setminus \{b_1,b_2,\ldots,b_k\}$ is adjacent to all elements of $A'$.
\end{lemma}

\begin{proof}
If there are an element $a_p\in A'$ having a neighbor $b\in B, b\neq b_p$ and an element $a_q\in A', (a_q,b)\not \in E(G)$, then
$b,b_p,a_p,a_q,b_q$ induce a $P_3+P_2$. Hence, every element of $N'$ must be adjacent to all elements of $A'$.
 \end{proof}

\begin{lemma} For each three vertices $a_1,a_2,a_3\in A$ such that $(a_1,a_2)\in E(G)$, $(a_1,a_3)\not \in E(G),(a_2,a_3)\not \in E(G)$,
we have $N_B(a_3)\subseteq N_B(a_1)\cup N_B(a_2)$.
If $D$ is a minimal subset of $A$ dominating $B$, then the graph $G[D]$ is complete multipartite.
\end{lemma}

\begin{proof} Assume that there is a vertex $b\in N_B(a_3)\setminus (N_B(a_1)\cup N_B(a_2))$. To avoid an induced $P_3+P_2$ in $G$,
each element of $N_B(a_1)\otimes N_B(a_2)$ is adjacent to $a_3$. Every element of $N_B(a_1)\cap N_B(a_2)$ is adjacent to $a_3$,
otherwise an element of the set, $a_1$, any element of $N_B(a_2)\setminus N_B(a_1), a_3$, and $b$ induce a $P_3+P_2$ in $G$.
We obtain that $N_B(a_1)\cup N_B(a_2)\subseteq N_B(a_3)$ which is impossible.

If $G[D]$ is not complete multipartite, then there are elements $a_1,a_2,a_3$ of $D$ such
that $(a_1,a_2)\in E(G),(a_1,a_3)\not \in E(G),(a_2,a_3)\not \in E(G)$. As $D$ is minimal,
then there is a vertex in $N_B(a_3)\setminus (N_B(a_1)\cup N_B(a_2))$ which is impossible.
\end{proof}\\

\begin{lemma}
Let ${\mathcal K}$ be the set of connected components of $G[A_1]$.
Then $A_1$ is independent or $\gamma''(G)=\min\limits_{K\in {\mathcal K}, N_B(V(K)\cup A_2\cup A_3)=B} \{\gamma''(G_K)|~G_K\triangleq G[V(K)\cup A_2\cup A_3\cup B]\}$.
\end{lemma}

\begin{proof} Let $D$ be a minimum subset of $A$ dominating $B$. We may assume that $D\cap A_1$ has at most
two elements and $A_1$ is not independent. By Lemma 5, $N_B(A_1)=N_B(K)$ for each connected component $K\in {\mathcal K}$ with at least two vertices.
If $(D\cap A_1)\setminus K$ has at least two elements, then $D\cap (A_2\cup A_3)\cup D\cap K\cup \{a',a''\}$ also dominates $B$ by Lemma 5,
where $a'$ and $a''$ are arbitrary adjacent vertices of $K$. If $(D\cap A_1)\setminus K$ has only one element $a^*$,
then  $D\cap K$ has an element $a^{**}$. The set $D\setminus \{a^*\}\cup \{a\}$ dominates $B$, where $a \in V(K)$ is
an arbitrary vertex adjacent to $a^{**}$. Hence, $\gamma''(G)$ must be equal to those minimum.
\end{proof}\\

According to the Lemma 6, we will assume that $G[A_1]$ and $G[A_2]$ are connected or one of them is connected and
the second is an empty graph or they are empty graphs.
By $\gamma''_k(G)$ we denote the size of a minimum subset of $A$ dominating $B$ and inducing a complete multipartite
subgraph with at most $k$ parts if one exists. If there is no such a subset, then $\gamma''_k(G)=+\infty$.

\begin{lemma}
For each fixed $k$, $\gamma''_k(G)$ can be computed in $O(|A|^k|V(G)|^{O(1)})$ time.
\end{lemma}

\begin{proof} Let $D$ be a minimal subset of $A$ dominating $B$. By Lemma 5, $G[D]$ is complete multipartite. By Lemma 4, any element of $B$ having a neighbor in a part of $G[D]$ must be adjacent to all elements of this part or to one its elements. If $G[D]$ has at most $k$ parts, then a subset $A'$ containing
exactly one element of each part and $N_B(A')$ can be removed from $G$ such that any element of $B^*$ in the resultant graph $G_{A'}(A^*,B^*)$ has only one neighbor in $A^*$. A subset $A'$ of
this type is said to be \emph{admissible}. If there are no admissible sets, then $\gamma''_k(G)=+\infty$. Otherwise, $\gamma''_k(G)$
is equal to the minimal of the sums $|A'|+|B^*|$ over all admissible subsets $A'$. This optimal sum can be computed in $O(|A|^k|V(G)|^{O(1)})$ time.
\end{proof}

\subsection{The classes $Free(\{P_3+P_2,orb\}), Free(\{P_3+P_2,K_5\}), Free(\{P_3+P_2,gem\})$, and $Free(\{P_3+P_2,sinker\})$}

\begin{lemma}
If $G(A,B)$ is an irreducible $\{P_3+P_2,orb\}$-free or $\{P_3+P_2,K_5\}$-free graph, then $\gamma''(G)=\gamma''_4(G)$.
\end{lemma}

\begin{proof} Let $D$ be a minimum subset of $A$ dominating $B$. For each $a\in D$, there is a vertex $b_a\in N_B(a)\setminus
\bigcup\limits_{v\in D\setminus \{a\}} N_B(v)$. Hence, $G[D]$ is complete multipartite with at most four parts by Lemma 5.
Therefore, $\gamma''(G)=\gamma''_4(G)$.
\end{proof}

\begin{lemma}
The dominating set problem for $\{P_3+P_2,gem\}$-free graphs can be polynomially reduced to the same problem for $\{P_5,gem\}$-free graphs
\end{lemma}

\begin{proof} Let $G(A,B)$ be an irreducible $\{P_3+P_2,gem\}$-free graph. If $G[A_1\cup A_2]$ is bipartite, then $\gamma''(G)=\gamma''_2(G)$. Hence, $\gamma''(G)$ can be computed in polynomial time. By Lemma 6, we will consider that $G[A_1]$ is connected and $G[A_2]$ is connected or an empty graph.
Hence, by Lemma 3, $G[A_1\cup A_3]$ is connected.

Let $b$ be a vertex of $B$ adjacent to a vertex in $A_1\cup A_3$ and to
a vertex in $A_2$. Then $A_3=\emptyset$ by Lemma 3 and the fact of $gem$-freeness of $G'$.
The vertex $b$ must be adjacent to all elements of $A_1$. Indeed, if there are vertices
$a^1_1,a^2_1\in A_1$ and $a_2\in A_2$ such that $(b,a^1_1)\in E(G),(b,a_2)\in E(G),(b,a^2_1)\not \in E(G)$, then
$b,a^1_1,a^2_1,a_2,x$ induce a $gem$ in $G'$. The set $A_2$ is independent, otherwise $G[A_2]$ is connected and $b$ must
be adjacent to all vertices of $A$. Moreover, $G[A_1]$ has at least two vertices $a'$ and $a''$, otherwise $G[A_1\cup A_2]$ is
bipartite. There is an element of $N_B(a')\setminus N_B(a'')$, and it can be adjacent to none element of $A_2$. Hence,
any subset of $A$ dominating $B$ must contain an element of $A_1$. Hence, $\gamma''(G)=\gamma''(G[A_1\cup A_3\cup N_B(A_1\cup A_3)])+\gamma''(G[A_2\cup (N_B(A_2)\setminus N_B(A_1\cup A_3))])$. This equality
also holds if there is no an element of $B$ adjacent to an element of $A_1\cup A_3$ and an element of $A_2$.

It is well-known that any $P_4$-free graph $H$ with at least two vertices
can be represented as follows. There are induced subgraphs $H_1,\ldots,H_k$ of $H$ such that $H=H_1+\ldots+H_k$ or
$H=H_1\times\ldots\times H_k$. Hence, if $H$ is a connected graph different from a complete graph and $IS_H$ is its maximum independent set, then $|IS_H|>1$ and any
element of $IS_H$ is adjacent to any element of $V(H)\setminus IS_H$. A maximum independent set of a $P_4$-free
graph can be computed in polynomial time.

Both graphs $G[A_1\cup A_3]$ and $G[A_2]$ are $P_4$-free, otherwise $G'$ is not $gem$-free. Assume that $G[A_1\cup A_3]$ is not complete. Then its vertex set can be decomposed into a maximum independent set $IS$ and the remaining part $V'$. It is easy to see that any vertex in $B$ adjacent to
a vertex in $IS$ and a vertex in $V'$ must be adjacent to all vertices of $IS$. Moreover, as $|IS|>1$, then there is a
vertex in $N_B(IS)\setminus N_B(V')$. Hence, $\gamma''(G[A_1\cup A_3\cup N_B(A_1\cup A_3)])=\gamma''(G[IS\cup N_B(IS)])+\gamma''(G[V'\cup (N_B(V')\setminus N_B(IS))])$. Clearly, $\gamma''(G[IS\cup N_B(IS)])=\gamma''_1(G[IS\cup N_B(IS)])$ and it can be computed in polynomial time by Lemma 7. Thus, computing
$\gamma''(G[A_1\cup A_3\cup N_B(A_1\cup A_3)])$ and $\gamma''(G[A_2\cup (N_B(A_2)\setminus N_B(A_1\cup A_3))])$ can be polynomially reduced to the case, when $G[A_1\cup A_3]$ and $G[A_2]$ are complete. In this case, $G[A_1\cup A_3\cup N_B(A_1\cup A_3)]$ and $G[A_2\cup (N_B(A_2)\setminus N_B(A_1\cup A_3))]$ are $\{P_5,gem\}$-free.
\end{proof}

\begin{lemma}
The dominating set problem for $\{P_3+P_2,sinker\}$-free graphs can be solved in polynomial time.
\end{lemma}

\begin{proof}
Let $G(A,B)$ be an irreducible $\{P_3+P_2,sinker\}$-free graph. Clearly, at least one
of the sets $A_1$ and $A_2$ is independent. Let $A_2$ be independent. If $A_3\neq \emptyset$,
then $H\triangleq G[A_1\cup A_3]$ is connected by Lemma 3. If $A_3$ is empty and $A_1$ is independent,
then $G[A]$ is bipartite and $\gamma''(G)=\gamma''_2(G)$. If $A_3$ is empty and $A_1$ is not independent, then $H$ is connected by Lemma 6.
If the graph $G[A]$ is $K_5$-free, then $\gamma''(G)=\gamma''_4(G)$ by Lemma 5. Hence, by Lemma 7,
$\gamma''(G)$ can be computed in polynomial time. We will assume that $H$ is a connected graph containing
a $K_4$.

Let $Q$ be a maximum clique of $H$ and $|Q|\geq 4$. Any element of $V(H)\setminus Q$ adjacent to an element of $Q$
must have exactly $|Q|-1$ neighbors in $Q$, as $G'$ is $sinker$-free. Since $G'$ is $sinker$-free and $H$ is connected,
there are no elements of $V(H)\setminus Q$ that have no neighbors in $Q$. Hence, if $a_1$ and $a_2$ belong to $V(H)\setminus Q$, then
they are adjacent if and only if $N(a_1)\cap Q\neq N(a_2)\cap Q$. Thus, $H$ is a complete multipartite graph with at least four parts.

There is no a vertex in $B$ adjacent to an element of $A_1\cup A_3$ and to an element of $A_2$ simultaneously.
Indeed, if such a vertex $b$ and its neighbors $a_1\in A_1\cup A_3, a_2\in A_2$ exist, then there is a clique $Q'$
of $H$ with at least three vertices such that $a_1\in Q'$. By Lemma 3, $a_2$ is adjacent to all vertices of $Q'$.
To avoid an induced $sinker$ in $G$, the vertex $b$ must be adjacent to at least two vertices of $Q'$. Hence, $b,a_2$, two
vertices in $Q'\cap N(b)$, and $x$ induce a $sinker$ in $G'$. So, $\gamma''(G)=\gamma''(G[A_1\cup A_3\cup N_B(A_1\cup A_3)])+
\gamma''(G[A_2\cup N_B(A_2)])$.

Let us show that there is no an element of $B$ adjacent to vertices in distinct parts of $H$. Let $b'$ be a vertex of this type,
$a'_1$ and $a'_2$ be its neighbors in distinct parts of $H$. There is a clique $Q''$ of $H$ containing exactly one representative
of each part of $H$ that also contains $a'_1$ and $a'_2$. Clearly, $|Q''|\geq 4$. To avoid an induced $sinker$ in $G'$, $b'$ must be adjacent
to all elements of $Q''$. If two vertices $a'\not \in Q''$ and $a''\in Q$ belong to the same part of $H$, then $b$ must be adjacent to $a'$.
Otherwise, $b$, any two elements of $Q''\setminus\{a''\},a'$, and $x$ induce a $sinker$ in $G'$. Therefore, $b'$ must be adjacent to
all vertices of $A_1\cup A_3$ which is impossible. So, $\gamma''(G[A_1\cup A_3\cup N_B(A_1\cup A_3)])=\sum\limits_{i=1}^{k}\gamma''(G[V_i\cup N_B(V_i)])$, where
$V_1,\ldots,V_k$ are all parts of $H$. By Lemma 7, the sum and $\gamma''(G[A_2\cup N_B(A_2)])$ can be computed in polynomial time, as
$\gamma''(G[A_2\cup N_B(A_2)])=\gamma''_1(G[A_2\cup N_B(A_2)])$ and $\gamma''(G[V_i\cup N_B(V_i)])=\gamma''_1(G[V_i\cup N_B(V_i)])$ for each $i$.
\end{proof}

\subsection{The classes $Free(\{P_3+P_2,K_{1,4}\}),Free(\{P_3+P_2,fork\}),Free(\{P_3+P_2,cricket\}),Free(\{P_3+P_2,bull\}),Free(\{P_3+P_2,kite\}),Free(\{P_3+P_2,dart\})$}

\begin{lemma}
The dominating set problem for $Free(\{P_3+P_2,K_{1,4}\})$ can be polynomially reduced to the same problem for $Free(\{P_5,K_{1,4}\})$.
\end{lemma}

\begin{proof}
Let $G(A,B)$ be an irreducible $\{P_3+P_2,K_{1,4}\}$-free graph. Let us show that $G^*$ is $\{P_5,K_{1,4}\}$-free. Its $P_5$-freeness is obvious. Suppose that $G^*$ has a $K_{1,4}$ induced by vertices $a,b_1,b_2,b_3,b_4$, where $(a,b_1),(a,b_2),(a,b_3),(a,b_4)$ are edges of this $K_{1,4}$.
Clearly, $a\in A$. There are at least three vertices among $b_1,b_2,b_3,b_4$ belonging to $B$. This three vertices, $a$, $x$ or $y$ induce a $K_{1,4}$ in $G'$. We have a contradiction.
\end{proof}

\begin{lemma}
The dominating set problem for $Free(\{P_3+P_2,fork\})$ can be polynomially reduced to the same problem for $Free(\{P_5,fork\})$.
\end{lemma}

\begin{proof}
Let $G(A,B)$ be an irreducible $\{P_3+P_2,fork\}$-free graph. To prove the lemma, we only need
to show that $G^*$ is $fork$-free. Suppose that $G^*$ has a $fork$ induced by
vertices $x_1,x_2,x_3,y_1,y_2$, where $(x_1,y_1),(x_2,y_1),(y_1,y_2),\\ (y_2,x_3)$ are edges of the $fork$. Clearly, $x_1,x_2,x_3\in B$ and $y_1,y_2\in A$.
The graph $G$ must have the edge $(y_1,y_2)$, otherwise $x_1,x_2,x_3,y_1,y_2$ induce a $P_3+P_2$. Then $G$ is not $fork$-free. We have a contradiction.
\end{proof}

\begin{lemma}
For every of the classes $Free(\{P_3+P_2,cricket\})$ and $Free(\{P_3+P_2,bull\})$, the dominating
set problem can be polynomially reduced to the same problem for $Free(\{P_5,fork\})$.
\end{lemma}

\begin{proof}
Let ${\mathcal X}$ be one of the two classes, $G(A,B)$ be an irreducible graph in ${\mathcal X}$.
Let $a_1$ and $a_2$ be elements of $A$ having a common neighbor $b\in B$. Taking into account that there are
vertices $b'\in N_B(a_1)\setminus N_B(a_2)$ and $b''\in N_B(a_2)\setminus N_B(a_1)$, it
is easy to see that $a_1$ and $a_2$ belong to exactly one of the sets $A_1,A_2,A_3$ by Lemma 3.
Hence, $\gamma''(G)=\gamma''(G_1)+\gamma''(G_2)+\gamma''(G_3)$, where $G_i$ is a subgraph of $G$ induced by $A_i\cup N_B(A_i)$.
Similar to the reasonings of the previous lemma, it is easy to check that all graphs $G^*_1,G^*_2,G^*_3$ are all $\{P_5,fork\}$-free. So, the lemma holds.
\end{proof}

\begin{lemma}
For every of the classes $Free(\{P_3+P_2,kite\})$ and $Free(\{P_3+P_2,dart\})$, the dominating
set problem can be polynomially reduced to the same problem for $Free(\{P_5,kite\})$ and $Free(\{P_5,dart\})$, respectively.
\end{lemma}

\begin{proof}
Let $G$ be an irreducible $\{P_3+P_2,kite\}$-free or $\{P_3+P_2,dart\}$-free graph. We
will show that $G[A_1\cup A_3]$ and $G[A_2\cup A_3]$ are $P_3$-free. Assume that $G[A_1\cup A_3]$ contains vertices $a_1,a_2,a_3\in A$ such that $(a_1,a_2)\in E(G),(a_2,a_3)\in E(G)$, and $(a_1,a_3)\not \in E(G)$.
Clearly, $N_B(a_1)\cap N_B(a_2)=\emptyset$. Otherwise, an element of $N_B(a_1)\setminus N_B(a_2)$, and an element of $N_B(a_1)\cap N_B(a_2),
a_1,a_2,x$ and an element of $N_B(a_1)\cap N_B(a_2),a_1,a_2,x,y$ induce a $dart$ and a $kite$ in $G'$, respectively.
If $A_3$ is non-empty, then $A_1\cup A_3$ and $A_2\cup A_3$ must be cliques by Lemma 3. Moreover, by Lemma 3, $A$ must be a clique,
and $G$ is $\{P_5,dart\}$-free or $\{P_5,kite\}$-free.

Let $A_3$ be empty. If $G[A]$ is bipartite, then $\gamma''(G)=\gamma''_2(G)$ and $\gamma''(G)$ can be computed in polynomial time.
Otherwise, by Lemma 6, we may assume that $G[A_1]$ is a clique and $G[A_2]$ is a clique or an empty graph. If $G[A_2]$ is a clique,
then $G$ is $\{P_5,dart\}$-free or $\{P_5,kite\}$-free by Lemma 3. If $A_2$ is independent and $\min(|A_1|,|A_2|)=1$, then $G[A]$ is also bipartite or a clique. Assume
that $A_1$ and $A_2$ have at least two vertices and $A_2$ is independent.

Let us show that there is no a vertex in $B$ adjacent to a vertex in $A_1$ and a vertex in $A_2$. Let $b$ be such a vertex.
If $G$ is $dart$-free, then $b$ is adjacent to all vertices of $A_1$. Then, for each $a',a''\in A_1$, an element of $N_B(a')\setminus N_B(a'')$,
$b,a',a'',x$ induce a $dart$ in $G'$. If $G$ is $kite$-free, then there are vertices $a_1,a_2\in A_1$ such that $(a_1,b)\not \in E(G), (a_2,b)\in E(G)$.
The set $A_2$ contains an element $a_3$. Clearly, $N_B(a_1)\cap N_B(a_2)=\emptyset$, otherwise an element of $N_B(a_1)\cap N_B(a_2),a_1,a_2,x,y$ induce a $kite$ in $G'$. Hence, an element of $N_B(a_1)\setminus N_B(a_3),a_1,a_2,a_3,b$ induce a $kite$ in $G$.

So, $\gamma''(G)=\gamma''(G[A_1\cup N_B(A_1)])+\gamma''(G[A_2\cup N_B(A_2)])$ and $\gamma''(G[A_2\cup N_B(A_2)])=\gamma''_1(G[A_2\cup N_B(A_2)])$. Hence, by Lemma 7,
computing $\gamma''(G)$ can be polynomially reduced to computing $\gamma''(G[A_1\cup N_B(A_1)])$, where $G[A_1\cup N_B(A_1)]$ is $\{P_5,dart\}$- or $\{P_5,kite\}$-free.
\end{proof}

\subsection{The class $Free(\{fork,K_3+K_2\})$}

Two non-adjacent vertices $x$ and $y$ of a graph are said to be \emph{quasi-twins} if $N(x)\subseteq N(y)$.
If $x$ and $y$ are quasi-twins of a graph $G$, then $\gamma(G)=\gamma(G\setminus \{y\})$. Hence,
the dominating set problem in a hereditary class can be polynomially reduced to the same problem for
its graphs without quasi-twins.

\begin{lemma}
Let $G$ be a connected $\{fork,K_3+K_2\}$-free graph without quasi-twins, and let $G\not \in Free(\{P_5\})$. Let $P\triangleq (x_1,x_2,\ldots,x_k)$ be a maximum
induced path of $G$ if $G\in Free(\{P_7\})$, otherwise let $P$ be a maximal induced path of $G$ with at least seven vertices. If a vertex $x\in V(G)\setminus V(P)$ has a neighbor $y\not \in \bigcup\limits_{v\in V(P)} N(v)$, then $x$ must be adjacent to all vertices of $P$.
\end{lemma}

\begin{proof}
Assume the opposite. The path $P$ must have at least five vertices. The vertex $x$ cannot have exactly one neighbor in $V(P)$, otherwise this
neighbor must be an end of $P$, and $P$ is not maximal. If $x$ has more than two neighbors in $V(P)$, then they must be consecutive in $P$ to
avoid an induced $fork$. Nevertheless, $G$ contains an induced $fork$, as $x$ cannot be adjacent to all vertices of $P$. Hence, $x$ must
have exactly two neighbors on $P$. They must be adjacent, as $G$ is $fork$-free. Moreover, $k\leq 6$, as $G$ is $K_3+K_2$-free. Hence, $P$ is maximum.
We may assume that $x_2$ and $x_3$ are the neighbors in the case $k=5$, $x_3$ and $x_4$ are the neighbors for $k=6$, since $G$ is $K_3+K_2$-free.
Suppose that $k=5$. As $G$ has no quasi-twins, there is a vertex $x'\in N(x_5)\setminus N(x_3)$. As $P$ is maximum,
$x'$ must have a neighbor in $V(P)\setminus \{x_5\}$. As $G$ is $\{fork,K_3+K_2\}$-free, $N(x')\cap V(P)=\{x_1,x_5\}$ or
$N(x')\cap V(P)=\{x_1,x_2,x_5\}$ or $N(x')\cap V(P)=\{x_1,x_2,x_4,x_5\}$ or $N(x')\cap V(P)=\{x_1,x_4,x_5\}$. Hence, $x'$ and $y$ cannot be adjacent.
Due to the maximality of $P$, $x'$ must be adjacent to $x$ in the case, when $N(x')\cap V(P)=\{x_1,x_2,x_5\}$. It is also true
in all three remaining cases, as $G$ is $K_3+K_2$-free. Hence, $G$ contains an induced $fork$. We have a contradiction. The case $k=6$ can be considered similarly.
\end{proof}

\begin{lemma} The dominating set problem for $\{fork,K_3+K_2\}$-free graphs can be polynomially reduced to the
same problem for $\{P_5,fork,K_3+K_2\}$-free graphs
\end{lemma}

\begin{proof}
Let $G$ be a connected $\{fork,K_3+K_2\}$-free graph without quasi-twins containing an induced $P_5$. Let $P\triangleq (x_1,\ldots,x_k)$ be a maximum
induced path of $G$ if $G\in Free(\{P_7\})$, otherwise let $P$ be a maximal induced path of $G$ with at least seven vertices. It can be computed in polynomial time. Assume that $V(P)$ is a dominating set of $G$. If $|V(P)|\leq 8$, then $\gamma(G)\leq 8$.
Suppose that $|V(P)|\geq 9$ and $G$ is distinct to a simple path and a cycle. Hence, there is a vertex $x\in V(G)\setminus V(P)$. Since $G$ is $fork$-free and
$P$ is maximal, $x$ has at least two neighbors on $P$. If it has exactly two neighbors, then $x$ must be adjacent to the ends of $P$.
As $G$ is $\{fork,K_3+K_2\}$-free and it is not a cycle, an element of $V(G)\setminus V(P)$ has at least three neighbors on $P$.
We may assume that $x$ has at least three neighbors on $P$. Let $x_s$ be the first neighbor of $x$ counting from $x_1$. Clearly,
$s\leq 2$, otherwise $x,x_s,x_{s+1},x_{s-1},x_{s-2}$ or $x,x_s,x_{s+1},x_1,x_2$ induce a $fork$ or a $K_3+K_2$, respectively. If $s=2$, then
$N(x)\cap \{x_4,\ldots,x_k\}$ has at most two vertices, and they must be adjacent. It is impossible. If $s=1$ and $(x,x_2)\not \in E(G)$,
then $N(x)\cap \{x_4,\ldots,x_k\}$ is a clique with at most two vertices. It is also impossible. If $(x,x_1)$ and $(x,x_2)$ are edges of $G$, then
no two vertices of $V(P)\setminus (N(x)\cup \{x_3\})$ are adjacent. Hence, this set has at most two elements. In other words, $x$ is adjacent to at least
$|V(P)|-3$ vertices of $P$. Therefore, each element of $V(G)\setminus V(P)$ is either only adjacent to the ends of $P$ or to at least
$|V(P)|-3$ its vertices. Hence, $\{x_1,x_2,x_3,x_4,x\}$ is a dominating set of $G$.

Now, assume that $V(P)$ is not a dominating set of $G$ and $\gamma(G)\geq 11$. Let $V_1$ be the set of those elements in $\bigcup\limits_{v\in V(P)} N(v)\setminus V(P)$ that are adjacent to all vertices of $P$ and have a neighbor outside $\bigcup\limits_{v\in V(P)} N(v)$. By the previous lemma, $V_1$ is not empty. Let $V_2$ be the set of those elements in $\bigcup\limits_{v\in V(P)} N(v)\setminus V(P)$ that are not adjacent to all vertices of $P$, $V_3\triangleq V(G)\setminus \bigcup\limits_{v\in V(P)} N(v)$. By the previous lemma, none element of $V_2$ has a neighbor outside $\bigcup\limits_{v\in V(P)} N(v)$. It is easy to check that each element of $V_1$ is adjacent to every element of $V_2$. The set $V_3$ must be non-empty, otherwise an element of $V_1$ and an element of $V(P)$ constitute a dominating set of $G$. As $G$ is connected and $fork$-free, any element of $V_3$ has a neighbor in $V_1$.

Let $H$ be a graph obtained from $G$ by removing any vertex of $P$. We will show that there is a minimum dominating set of $H$ containing an element of $V_1$.
Hence, this set must be a dominating set of $G$. Therefore, $\gamma(H)=\gamma(G)$. Let $D$ be a minimum dominating set of $H$ containing no elements of $V_1$.
The set $D\cap \bigcup\limits_{v\in V(P)} N(v)$ has at most one element, otherwise any element of $V(P)$, any element of $V_1$, and $V_3\cap D$ form a dominating set of $H$. To avoid an induced $fork$ in $H$, $N(z')\cap V_1=N(z'')\cap V_1$ for any two vertices $z'\in V_3$ and $z''\in V_3$. We may consider that any element
$z^*\in D\cap V_3$ has a neighbor in $N(z^*)\setminus \bigcup\limits_{z\in D\cap A_3\setminus \{z^*\}} N(z)$, otherwise $D\setminus \{z^*\}\cup \{y^*\}$ is
a minimum dominating set of $G$, where $y^*\in V_1$ is an arbitrary neighbor of $z^*$. Indeed, any neighbor of $z^*$ in $V_3$ must be adjacent to $y^*$.
As $\gamma(G)\geq 11$, then $|V_3\cap D|\geq 9$. Let $V_3\cap D\triangleq \{z_1,\ldots,z_p\}$. For every $i$, there is a vertex $y_i\in V_1$ such that $y_i\in N(z_i)\setminus \bigcup\limits_{j=1,j\neq i}^{p} N(z_j)$. Hence, $V_3\cap D$ is independent. By Ramsey's theorem, among $y_1,\ldots y_p$ some three vertices
are pairwise non-adjacent or some four vertices form a clique. The first alternative is impossible, as $H$ is $fork$-free. Suppose that $y_1,y_2,y_3,y_4$ constitutes a clique of $H$. Let us show that $(D\setminus \{z_1,z_2,z_3,z_4\})\cup \{y_1,y_2,y_3,y_4\}$ is a dominating set of $H$. Clearly $V_3\cap \bigcup\limits_{i=1}^{4} N(z_i)\subseteq \bigcup\limits_{i=1}^{4} N(y_i)$. If $y\in V_1\setminus \{y_1,\ldots,y_p\}$ has a neighbor in $\{z_1,z_2,z_3,z_4\}$, then $y$ must have a neighbor in $\{y_1,y_2,y_3,y_4\}$, otherwise $H$ is not $K_3+K_2$.

So, deleting vertices of long induced paths in $\{fork,K_3+K_2\}$-free graphs gives a polynomial reduction to $\{P_5,fork,K_3+K_2\}$-free graphs.
\end{proof}

\section{Main result}

The following result was proved in \cite{M15_2}.

\begin{lemma}
The dominating set problem for a hereditary class ${\mathcal X}\subseteq Free(\{G+O_1\})$ can be polynomially
reduced to the same problem for ${\mathcal X}\cap Free(\{G\})$.
\end{lemma}

Recall that the classes ${\mathcal S},{\mathcal T},{\mathcal Q},{\mathcal Q^*}$ defined in the third section are boundary for the dominating set problem.

\begin{thm}
Let ${\mathcal X}$ be defined by a set of forbidden induced subgraphs with at most five vertices. The dominating set problem
for ${\mathcal X}$ is NP-complete if it includes at least one of the classes ${\mathcal S},{\mathcal T},
{\mathcal Q}$. Otherwise, the problem can be solved in polynomial time for ${\mathcal X}$.
\end{thm}

\begin{proof} Let ${\mathcal Y}$ be a minimal set such that ${\mathcal X}=Free({\mathcal Y})$. By Theorem 1,
the dominating set problem is NP-complete for ${\mathcal X}$ if it includes ${\mathcal S}$ or ${\mathcal T}$ or
${\mathcal Q}$. Assume that ${\mathcal S}\nsubseteq {\mathcal X},{\mathcal T}\nsubseteq {\mathcal X},{\mathcal Q}\nsubseteq {\mathcal X}$.
Hence, ${\mathcal Y}$ contains a forest. If ${\mathcal Y}$ contains an induced subgraph of $P_5+O_5$, then ${\mathcal X}$
is easy for the problem by Theorem 2 and Lemma 17. Suppose that $P_3+P_2$ belongs to ${\mathcal Y}$. Let $G$ be a graph in ${\mathcal Q}$ containing at most five vertices. The graph $G$ cannot contain $C_4,C_5,2K_2$, and the complements of $K_2+O_3,K_3+O_2$ as an induced subgraph. Taking into account a list of all five-vertex graphs, it is easy to check that $G$ is an induced subgraph of one of the graphs $P_4+O_5,K_5+O_5,orb+O_5,sinker+O_5,kite+O_5,dart+O_5,cricket+O_5,fork+O_5,K_{1,4}+O_5,gem+O_5,bull+O_5$.
By Lemmas 8-14 and 17, the problem for ${\mathcal X}$ can be polynomially reduced to the same problem for
the classes $Free(\{P_5,G\})$ and $Free(\{P_5,fork\})$. Hence, by Theorem 2, ${\mathcal X}$ is easy.
Assume that $fork\in {\mathcal Y}$ and ${\mathcal Y}$ does not contain a subgraph of $P_5$. Hence,
${\mathcal Y}$ contains a graph $H\in {\mathcal T}$, which is a subgraph of a $hammer$. The classes
$Free(\{fork,bull\})$ and $Free(\{fork,hammer\})$ are easy for the problem \cite{M15_2}. By these facts,
Lemmas 16 and 17, and Theorem 2, the problem is polynomial-time solvable for ${\mathcal X}$.
\end{proof}\\

There is an interesting detail concerning the previous theorem. Namely, textual replacing five to six
leads to an incorrect statement. Indeed, none of the classes  ${\mathcal S},{\mathcal T},{\mathcal Q}$ is
contained in $Free(\{K_{1,4},P_3+P_3\})$. Hence, it should be an easy case for the dominating set problem
assuming the correctness of those fact. But, by Theorem 1, the class is a hard case for the problem, since
$Free(\{K_{1,4},P_3+P_3\})\supseteq {\mathcal Q^*}$.

\section*{Acknowledgements}

The article was prepared within the framework of the Academic Fund Program at the
National Research University Higher School of Economics (HSE) in 2015--2016 (grant
15-01-0010)  and supported within the framework of a subsidy granted to the HSE by
the Government of the Russian Federation for the implementation of the Global
Competitiveness Program.

\end{document}